\documentclass[10pt, a4paper]{amsart}
\usepackage{amsmath, amsfonts, amssymb, mathtools, url, graphicx, graphicx}
\usepackage{newtxtext}
\usepackage{newtxmath}

\usepackage{enumitem}
\usepackage[usenames, dvipsnames]{color}

\allowdisplaybreaks
\usepackage{tikz}
\usetikzlibrary{automata,positioning}
\usepackage{caption}
\usepackage{subcaption}
\usepackage{algorithm,algorithmic}
\usepackage{multicol}

\newtheorem{proposition}{Proposition}
\newtheorem{defn}{Definition}

\newtheorem{remark}{Remark}
\newtheorem{example}{Example}

\newcommand{\abs}[1]{\left\lvert{#1}\right\rvert}

\newcommand{\pmat}[1]{\begin{pmatrix}#1\end{pmatrix}}

\newcommand{\R}{\mathbb{R}}
\newcommand{\N}{\mathbb{N}}
\renewcommand{\P}{\mathcal{P}}
\newcommand{\A}{\mathcal{A}}
\newcommand{\M}{\mathcal{M}}
\newcommand{\D}{\mathcal{D}}
\newcommand{\row}{\mathrm{row}}

\title{Learning switched systems from simulation experiments}
\author{Atreyee Kundu}
\thanks{The author is with the Department of Electrical Engineering, Indian Institute of Science Bangalore, Bengaluru - 560012, Karnataka, India, E-mail: \texttt{atreyeek@iisc.ac.in}}
\thanks{The author thanks Pavithra Prabhakar for helpful discussions and an introduction to the \(L^*\)-algorithm.}


\date{\today}

\begin{document}

	\begin{abstract}
		The design of decision and control strategies for switched systems typically requires complete knowledge of (i) mathematical models of the subsystems and (ii) restrictions on admissible switches between the subsystems. We propose an active learning algorithm that infers (i) and (ii) for discrete-time switched systems whose subsystems dynamics are governed by sets of scalar polynomials and switching signals are constrained by automata. We collect data from gray-box simulation models of the switched systems for this purpose. Our technique for learning (i) involves linear algebraic tools, while for learning (ii) we employ a modified version of the well-known \(L^*\)-algorithm from machine learning literature. A numerical example is presented to demonstrate our learning algorithm.
	\end{abstract}

    \maketitle



\section{Introduction}
\label{s:intro}
     A switched system contains several subsystems and a switching signal that governs, at every instant of time, which subsystem is active at that instant \cite[\S 1.1.2]{Liberzon}. Such systems find wide applications in power systems and power electronics, aircraft and air traffic control, network and congestion control, etc. \cite[p.\ 5]{Sun}. The knowledge of mathematical models of the subsystems (e.g., transfer functions, state-space models, or kernel representations) and restrictions on the set of admissible switching signals (e.g., admissible switches between the subsystems, admissible dwell times on the subsystems, etc.) is a key requirement for the design of decision and control algorithms for switched systems. As a result, system identification techniques for these systems constitute a key topic in the literature, see e.g., the survey paper \cite{Garulli2012}, the recent works \cite{Petreczky2018,Gosea2018,Dabbene2019} and the references therein. In general, the problem of identification of switched systems is known to be NP-hard \cite{Lauer2016} and the most widely used techniques to address it include (a) algebraic methods \cite{Vidal2003}, (b) mixed integer programming methods \cite{Roll2004}, (c) clustering based methods \cite{Ferrari-Trecate2003}, (d) Bayesian methods \cite{Juloski2005}, (e) bounded error identification methods \cite{Bemporad2005}, and (f) sparse optimization methods \cite{Bako2011}. In this paper we explore a new tool, viz., active learning techniques for the identification of discrete-time switched systems.

     In modern industrial setups, simulation models are often provided by the system manufacturers to study the system behaviour with respect to various sets of inputs prior to their application to the actual system.  However, the mathematical models of the subsystems and the constraints on their operations underlying this simulation model are often not explicitly known to the user. We address the problem of \emph{learning} a switched system by collecting data from such a simulation model. The said learning task involves inferring two ingredients: (i) mathematical models of the subsystems and (ii) restrictions on the set of admissible switching signals.

     We consider discrete-time switched systems whose subsystems dynamics are governed by sets of scalar polynomials and switching signals are constrained by automata (henceforth to be called as \emph{restriction automata}).\footnote{The reader is referred to \cite{Hopcroft} for details on automata theory.} A restriction automaton is a strongly connected, directed and labelled multigraph with a unique initial node.\footnote{Recall that a directed multigraph is a directed graph that can have both self-loops and parallel edges \cite[p.\ 7]{Bollobas}.} The edges of a restriction automaton are labelled with indices of the subsystems. A switching signal admissible on a restriction automaton is an infinite sequence of subsystems indices such that the automaton admits an infinite path whose edge labels (in order) match the sequence. The language of the restriction automaton is a finite set of finite sequences of subsystems indices that captures properties of the admissible switching signals. Given the number of subsystems, the dimension of the subsystems, the order of the set of scalar polynomials governing the subsystems dynamics (henceforth to be called as \emph{component polynomials}), the maximum length of the elements in the language of the restriction automaton, and a simulation model of the switched system from which certain information about (a) state trajectories of the individual subsystems and (b) existence of finite sequences of subsystems indices in the language of the restriction automaton can be obtained, we present an algorithm that learns two components: (i') the state-space models of the subsystems and (ii') a minimal automaton that accepts the language of the restriction automaton.\footnote{By ``minimal'', we mean the automaton with the minimum number of nodes that accepts the language under consideration.} Our simulation model of a switched system is a gray-box because partial information about the switched system, as described above, is available.

     Towards inferring (i'), we rely on linear algebraic techniques. The Learner carries out the following tasks: (a) She performs simulation experiments to collect finite traces of state trajectories for the individual subsystems with respect to sets of initial values of the trajectories that satisfy certain properties; (b) She determines the coefficients of the component polynomials by solving sets of systems of linear equations. Our learning process for (ii') is a modified version of the \(L^*\)-algorithm \cite{Angluin1987} from machine learning literature. The following steps are executed: (a) The Learner determines if a finite sequence of subsystems indices is an element of the language of the restriction automaton or not by performing simulation experiments; (b) She conjectures an automaton when the above information is collected for a set of finite sequences of subsystems indices that satisfies certain conditions; (c) Correctness of the conjectured automaton is determined by matching its language with the language of the restriction automaton. If a sequence that is misclassified by the conjectured automaton (henceforth to be called as a \emph{counter-example}) is obtained, then the Learner updates the conjectured automaton by collecting more information from the simulation model, and repeats the check for the correctness of the conjectured automaton;\footnote{A finite sequence of subsystems indices is misclassified by the conjectured automaton, if the sequence is an element of the language of the conjectured automaton but not an element of the language of the restriction automaton, or vice-versa.} (d) The learning process terminates when a counter-example no longer exists. The algorithm proposed in this paper falls in the domain of \emph{active learning} because the Learner has access to a simulation model that aids the learning process. The running time of our learning algorithm is polynomially bounded. We provide an example to demonstrate our learning algorithm.

     The remainder of this paper is organized as follows: in \S\ref{s:prob_stat} we describe the problem under consideration. Our learning algorithm appears in \S\ref{s:mainres}. We conclude in \S\ref{s:concln} with a brief discussion on future research directions.

     {\bf Notation}. \(\R\) is the set of real numbers and \(\N\) is the set of natural numbers, \(\N_{0} = \N\cup\{0\}\). For a finite set \(S\), we denote by \(\abs{S}\) the size of \(S\).
\section{Problem Statement}
\label{s:prob_stat}
    We consider a set of systems
    \begin{align}
    \label{e:family}
        x(t+1) = f_{p}(x(t)),\:\:x(0) = x_{0},\:\:p\in\P,
    \end{align}
    where \(x(t)\in\R^{d}\) is the vector of states at time \(t\), \(\P = \{1,2,\ldots,N\}\) is an index set, and the functions \(f_{p}:\R^{d}\to\R^{d}\) take the form
        \(f_{p}(\xi) = \pmat{\displaystyle{\sum_{k=0}^{m}a_{p,1k}\xi_{1}^{k}}&
        \displaystyle{\sum_{k=0}^{m}a_{p,2k}\xi_{2}^{k}}&
        \cdots &
        \displaystyle{\sum_{k=0}^{m}a_{p,dk}\xi_{d}^{k}}}^\top\), \(p\in\P\),
    where \(m\in\N\), \(a_{p,ik}\in\R\), \(i=1,2,\ldots,d\), \(k=0,1,\ldots,m\), and \(\xi_{i}\) is the \(i\)-th element of the vector \(\xi\), \(i=1,2,\ldots,d\). Let \(\sigma:\N_{0}\to\P\) be a switching signal. A switched system generated by the subsystems \(p\in\P\) and a switching signal \(\sigma\) is given by
    \begin{align}
    \label{e:swsys}
        x(t+1) = f_{\sigma(t)}(x(t)),\:\:x(0) = x_{0},\:\:t\in\N_{0}.
    \end{align}
    \begin{remark}
    \label{rem:linear}
    \rm{
        Notice that if \(m=1\) and \(a_{p,i0} = 0\) for all \(i=1,2,\ldots,d\) and \(p\in\P\), then \eqref{e:swsys} is a switched linear system whose subsystems are \(f_{p}(\xi) = A_{p}\xi\), \(p\in\P\) with \(A_{p}\in\R^{d\times d}\), \(p\in\P\) diagonal matrices.
    }
    \end{remark}

    Let the set of admissible switching signals of \eqref{e:swsys} be restricted by an automaton.
    \begin{defn}
    \label{d:automat}
    \rm{
        A \emph{restriction automaton} \(\A = (V,v_{0},E,\P,\ell)\) for a switched system \eqref{e:swsys} is a strongly connected, directed and labelled multigraph, where \(V\) is a finite set of nodes, \(v_{0}\in V\) is the unique initial node, \(E\) is a finite set of edges, \(\P\) is the index set described above, and \(\ell:E\to\P\) is an edge labelling function.
    }
    \end{defn}
    \begin{defn}
    \label{d:path}
    \rm{
        A \emph{path} on \(\A\) is a sequence of nodes starting from the initial node, \(W = \overline{v}_{0},\overline{v}_{1},\ldots,\overline{v}_{n}\), where \(\overline{v}_{0} = v_{0}\), \(\overline{v}_{k}\in V\) such that \((\overline{v}_{k-1},\overline{v}_{k})\in E\), \(k=1,2,\ldots,n\).
    }
    \end{defn}
    Notice that since \(\A\) is a multigraph, the choice of the edge \((\overline{v}_{k-1},\overline{v}_{k})\in E\), \(k\in\{1,2,\ldots,n\}\) is not necessarily unique. The \emph{length of a path} is one less than the number of nodes that appear in the sequence, counting repetitions. For example, the length of the path \(W\) in Definition \ref{d:path} is \(n\). By the term \emph{infinite path}, we mean a path of infinite length, i.e., it has infinitely many nodes.
    \begin{defn}
    \label{d:adm_sw}
    \rm{
        A switching signal \(\sigma = \sigma(0)\sigma(1)\sigma(2),\sigma(3)\)\(\cdots\) is \emph{admissible on \(\A\)}, if \(\A\) admits an infinite path \(W = \overline{v}_{0},\overline{v}_{1},\overline{v}_{2},\ldots\) such that \(\ell(\overline{v}_{k},\overline{v}_{k+1})=\sigma(k)\), \(k\in\N_{0}\).
    }
    \end{defn}
    Let \(S_{\A}\) denote the set of all switching signals admissible on \(\A\). By strong connectivity of \(\A\), the set \(S_{\A}\) is non-empty. Fix \(\sigma\in S_{\A}\). We let \(\sigma|_{\tau_{1}}^{\tau_{2}}\) denote the segment of \(\sigma\), \(\sigma|_{\tau_{1}}^{\tau_{2}} = \sigma(\tau_{1})\sigma(\tau_{1}+1)\cdots\sigma(\tau_{2})\), of length \(\tau_{2}-\tau_{1}+1\), \(\tau_{1},\tau_{2}\in\N_{0}\).
    \begin{defn}
    \label{d:lang}
    \rm{
        The \emph{language of \(\A\)} is given by
        \begin{align}
        \label{e:lang}
            L_{\A} = \{\sigma|_{0}^{\tau},\:\tau=0,1,\ldots,M,\:\text{for all}\:\sigma\in S_{\A}\}\cup\{\lambda\},
        \end{align}
        where \(M\in\N\) is a large number such that for all \(\sigma\in S_{\A}\), \(\sigma|_{\tau'_{1}}^{\tau'_{2}}\in L_{\A}\), \(\tau'_{1}\geq M\), \(\tau'_{2}-\tau'_{1}\leq M+1\), and \(\lambda\) denotes the empty sequence, i.e., a sequence of length \(0\).
    }
    \end{defn}
    Let \(\tilde{p}_{n}\) denote a sequence of subsystems indices of length \(n\), i.e., \(\tilde{p}_{n} = \tilde{p}_{n}(1)\tilde{p}_{n}(2)\cdots\tilde{p}_{n}(n)\) with \(\tilde{p}_{n}(k)\in\P\), \(k=1,2,\ldots,n\). We say that \(\tilde{p}_{n}\in L_{\A}\), if there exists \(\sigma\in S_{\A}\) such that \(\sigma|_{0}^{n-1}\) satisfies \(\sigma(k) = \tilde{p}_{n}(k+1)\), \(k=0,1,\ldots,n-1\).
    Let
    \(
        {\P}^*\) be the set of all sequences of subsystems indices, \(\tilde{p}_{n}\), of length \(n=0,1,\ldots,M\).

    Suppose that the Learner has the following set of information: a) the number of subsystems, \(N\), b) the dimension of the subsystems, \(d\), c) the order of the component polynomials, \(m\), and d) the maximum length, \(M\), of the elements of \(L_{\A}\). In addition, she has access to a gray-box simulation model, \(\M\), of the switched system \eqref{e:swsys}, containing two components:
    \begin{enumerate}[label = (\arabic*), leftmargin = *]
        \item \(\M_{\P}:\P\times\R^{d}\to\R^{d}\), which takes a pair \((p,\overline{x})\in\P\times\R^{d}\), as input and provides an output
        \(
            \M_{\P}(p,\overline{x}) = f_{p}(\overline{x})\in\R^{d},
        \)
        and
        \item \(\M_{\A}:\P^{*}\to\{0,1\}\), which takes a sequence of indices \(\tilde{p}_{n}\in\P^*\) as input and provides an output
            \begin{align*}
                \M_{\A}(\tilde{p}_{n}) =
                \begin{cases}
                    1,\:\:&\text{if}\:\tilde{p}_{n}\in L_{\A},\\
                    0,\:\:&\text{if}\:\tilde{p}_{n}\notin L_{\A}.
                \end{cases}
            \end{align*}
    \end{enumerate}
    We assume that the model, \(\M\), is accurate, and always provides correct information. The Learner's objective is to infer (i'') the coefficients \(a_{p,ik}\), \(i=1,2,\ldots,d\), \(k=0,1,\ldots,m\), \(p\in\P\) of the component polynomials and (ii'') a minimal automaton that accepts the language \(L_{\A}\) from the knowledge of \(N\), \(d\), \(m\), \(M\), and the information obtained from the simulation model, \(\M\).

    \begin{remark}
    \label{rem:sys-id_diff}
    \rm{
        In general, the problem of system identification for switched systems deals with inferring mathematical models of the subsystems and the discrete switching signal from sets of inputs (possibly controlled) - output (possibly noisy) data available to the user. In this paper we consider open-loop operations of a switched system \eqref{e:swsys} in the sense that there is no exogenous continuous control input applied and only a discrete switching signal controls the mode of operation of the system. We operate in the setting where initial values of the state trajectories for each subsystem can be chosen externally and the corresponding one-step evolution of these trajectories can be observed. The class of restriction automata under consideration expresses a large set of switching signals, viz., with no restrictions, with restrictions on admissible switches between  the subsystems and/or restrictions on admissible dwell times on the subsystems, etc.
    }
    \end{remark}

    We now describe our learning techniques.
\section{Results}
\label{s:mainres}
\subsection{Learning state-space models of subsystems from \(\M_{\P}\)}
\label{ss:subsys_learning}
    We begin with learning (i''). Algorithm \ref{algo:model_learn} computes the coefficients \(a_{p,ik}\), \(k=0,1,\ldots,m\), \(i=1,2,\ldots,d\) of the polynomials \(\displaystyle{\sum_{k=0}^{m}a_{p,ik}}\), \(i=1,2,\ldots,d\) by employing the following steps:
    First, \(\overline{x}_{0}\in\R^{d}\) is fixed to be a vector containing \(d\)-many zeros. The pair \((p,\overline{x}_{0})\) is input to \(\M_{\P}\) and the vector \(x'_{0} = \pmat{a_{p,10} & a_{p,20} & \cdots & a_{p,d0}}^\top\) is obtained. Second, for each \(k=1,2,\ldots,m\), \(\overline{x}_{k}\in\R^{d}\) is fixed to be a vector containing \(d\)-many \(k\)'s, the pair \((p,\overline{x}_{k})\) is input to \(\M_{\P}\) and a vector \(
                x'_{k} = \pmat{a_{p,10}+a_{p,11}\cdot k^{1} + \cdots + a_{p,1m}\cdot k^{m}\\
                               \vdots\\
                               a_{p,d0}+a_{p,d1}\cdot k^{1} + \cdots + a_{p,dm}\cdot k^{m}}
            \)
            is obtained.\footnote{For a vector \(x'_{k}\in\R^{d}\), \(x'_{k,i}\) denotes its \(i\)-th component, \(i=1,2,\ldots,d\), \(k=0,1,\ldots,m\).} Third, for each \(i=1,2,\ldots,d\), the system of linear equations \eqref{e:sys_lin} is solved, and the values of \(a_{p,ik}\), \(k=1,2,\ldots,m\) are obtained.
    \begin{algorithm}[htbp]
			\caption{Learning state-space models of the subsystems} \label{algo:model_learn}
		\begin{algorithmic}[1]
			\renewcommand{\algorithmicrequire}{\textbf{Input:}}
			\renewcommand{\algorithmicensure}{\textbf{Output:}}
			
			\REQUIRE The total number of subsystems, \(N\), the dimension of the subsystems, \(d\), the order of the component polynomials, \(m\), and a simulation model, \(\M_{\P}\).
            \ENSURE The functions \(f_{p}\), \(p\in\P\).

            \FOR {\(p=1,2,\ldots,N\)}
                \STATE Set \(\overline{x}_{0} = \pmat{0 & 0 &\cdots & 0}^\top\).
                \STATE Input \((p,\overline{x}_{0})\) to \(\M_{\P}\) and obtain \(x'_{0} = \M_{\P}(p,\overline{x})\).
                \FOR {\(i=1,2,\ldots,d\)}
                    \STATE Set \(a_{p,i0} = x'_{0,i}\).
                \ENDFOR
                \FOR {\(k=1,2,\ldots,m\)}
                    \STATE Set \(\overline{x}_{k} = \pmat{k & k & \cdots & k}^\top\in\R^{d}\).
                    \STATE Input \((p,\overline{x}_{k})\) to \(\M_{\P}\) and obtain \(x'_{k} = \M_{\P}(p,\overline{x}_{k})\).
                    \FOR {\(i=1,2,\ldots,d\)}
                        \STATE Set \(a_{p,i1}\cdot k^{1}+a_{p,i2}\cdot k^{2} + \cdots a_{p,im}\cdot k^{m} = x'_{k,i}-a_{p,i0}\).
                    \ENDFOR
                \ENDFOR
                \FOR {\(i=1,2,\ldots,d\)}
                    \STATE Solve the following system of linear equations for \(a_{p,ik}\in\R\), \(k=1,2,\ldots,m\):
                    \begin{align}
                    \label{e:sys_lin}
                        \hspace*{-0.8cm}a_{p,i1}\cdot 1^{1} + a_{p,i2}\cdot 1^{2} + \cdots + a_{p,im}\cdot 1^{m} &= x'_{1,i} - a_{p,i0}\nonumber\\
                        \hspace*{-0.8cm}a_{p,i1}\cdot 2^{1} + a_{p,i2}\cdot 2^{2} + \cdots + a_{p,im}\cdot 2^{m} &= x'_{2,i} - a_{p,i0}\nonumber\\
                        \vdots\\
                        \hspace*{-0.8cm}a_{p,i1}\cdot m^{1} + a_{p,i2}\cdot m^{2} + \cdots + a_{p,im}\cdot m^{m} &= x'_{m,i} - a_{p,i0}.\nonumber
                    \end{align}
                \ENDFOR
                \STATE Output \(f_{p}(x) = \pmat{\displaystyle{\sum_{k=0}^{m}a_{p,1k}x_{1}^{k}}&
        \cdots &
        \displaystyle{\sum_{k=0}^{m}a_{p,dk}x_{d}^{k}}}^\top\).
            \ENDFOR
		\end{algorithmic}
	\end{algorithm}
    \begin{proposition}
    \label{prop:unique_soln}
        Fix \(p\in\P\) and \(i\in\{1,2,\ldots,d\}\). The system of linear equations \eqref{e:sys_lin} has a unique solution.
    \end{proposition}
    \begin{proof}
        Let us rewrite \eqref{e:sys_lin} as \(AX = b\), where
        \(A = \pmat{1^1 & 1^2 & \cdots & 1^m\\2^1 & 2^2 & \cdots & 2^m\\\vdots & & \cdots & \vdots\\m^1 & m^2 & \cdots & m^m}\in\R^{m\times m}\), \(X = \pmat{a_{p,i1}\\a_{p,i2}\\\vdots\\a_{p,im}}\in\R^{m}\) and \(b = \pmat{x'_{1,i}-a_{p,i0}\\x'_{2,i}-a_{p,i0}\\\vdots\\x'_{m,i}-a_{p,i0}}\in\R^{m}\). Notice that \(A\) is non-singular because \(1,2,\ldots,m\) are distinct \cite[Fact 5.16.3]{Bernstein}. It follows at once that \eqref{e:sys_lin} admits a unique solution.
    \end{proof}

    \begin{proposition}
    \label{prop:bounded_time}
        Algorithm \ref{algo:model_learn} terminates in bounded time.
    \end{proposition}
    \begin{proof}
        Algorithm \ref{algo:model_learn} performs \(p(m+1)\)-many queries with \(\M_{\P}\) and solves the system of linear equations \eqref{e:sys_lin} \(pd\)-many times. In view of Proposition \ref{prop:bounded_time}, we have that for a fixed \(p\) and \(i\), the complexity of solving \eqref{e:sys_lin} is at most \(O(m^{3})\). It follows that the execution of Algorithm \ref{algo:model_learn} terminates in bounded time.
    \end{proof}

    \begin{remark}
    \label{rem:sys-id_nonlinear}
    \rm{
        While a vast body of system identification literature deals with switched linear systems, nonlinear subsystems are also considered earlier in a few instances, see e.g., the recent work \cite{Bianchi2018} and the references therein. On the one hand, the available results do not restrict the subsystems dynamics to be governed by sets of scalar polynomials as considered in this paper and propose techniques for inferring them from sets of input-output data. On the other hand, we present an active learning technique for inferring a class of switched systems whose subsystems dynamics are governed by sets of scalar polynomials.
    }
    \end{remark}

    We now move on to the learning of (ii'').
\subsection{Learning minimal automaton that accepts \(L_{\A}\) from \(\M_{\A}\)}
\label{ss:automat_learning}
    For a sequence of indices \(\tilde{p}_{n}\in\P^*\), its prefixes are the sequences \(\lambda\), \(\tilde{p}_{1}\), \(\tilde{p}_{1}\tilde{p}_{2},\ldots,\tilde{p}_{1}\tilde{p}_{2}\cdots\tilde{p}_{n-1}\), and its suffixes are the sequences \(\tilde{p}_{2}\), \(\tilde{p}_{2}\tilde{p}_{3},\ldots,\tilde{p}_{2}\tilde{p}_{3}\cdots\tilde{p}_{n},\lambda\).
    \begin{defn}
    \label{d:prefix+suffix_closed}
    \rm{
        A set \(\P'\subseteq\P^*\) is \emph{prefix-closed} if every prefix of every element \(\tilde{p}_{n}\in\P'\) is also an element of \(\P'\). The set \(\P'\) is \emph{suffix-closed} if every suffix of every element \(\tilde{p}_{n}\in\P'\) is also an element of \(\P'\).
    }
    \end{defn}
    In the sequel we will skip explicit reference to the length of an element of \(\P^*\) by omitting the subscript \(n\), whenever the length is not relevant or is immediate from the context. For two sequences of indices \(p',p''\in\P^*\), we let \(p'\cdot p''\) denote their concatenation. Further, for two sets \(\P',\P''\subseteq\P^*\), we let \(\P'\cdot\P''=\{p'\cdot p''\:|\:p'\in\P',\:p''\in\P''\}\).
    \begin{defn}
    \label{d:obsv_tab}
    \rm{
        An \emph{observation table} \((Q,R,T)\) consists of three components: (a) a non-empty finite prefix-closed set \(Q\subseteq\P^*\), (b) a non-empty finite suffix-closed set \(R\subseteq\P^*\), and (c) a function \(T:\bigl((Q\cup Q\cdot\P)\cdot R\bigr)\to\{0,1\}\), defined as
        \begin{align*}
            T(\tilde{p}) =
            \begin{cases}
                0,\:\:&\:\:\text{if}\:\M_{\A}(\tilde{p})=0,\\
                1,\:\:&\:\:\text{if}\:\M_{\A}(\tilde{p})=1,
            \end{cases}
            \:\:\tilde{p}\in\bigl((Q\cup Q\cdot\P)\cdot R\bigr).
        \end{align*}
    }
    \end{defn}
    An observation table can be visualized as a two-dimensional array with rows labelled by the elements of \((Q\cup Q\cdot\P)\) and columns labelled by the elements of \(R\), with the entry for row \(q\) and column \(r\) equal to \(T(q\cdot r)\). In the sequel we will denote by \(\row(\tilde{p})\) the row of the observation table labelled by the element \(\tilde{p}\in Q\), and \(\row(\tilde{p})\neq 0\) to denote that not all elements in the row of \((Q,R,T)\) under consideration are \(0\).
    \begin{defn}
    \label{d:closed}
    \rm{
        An observation table \((Q,R,T)\) is \emph{closed} if for each \(p'\in Q\cdot\P\), there exists \(p''\in Q\) such that \(\row(p') = \row(p'')\).
    }
    \end{defn}
    \begin{defn}
    \label{d:consistent}
    \rm{
        An observation table \((Q,R,T)\) is \emph{consistent} if the following condition holds: whenever \(p',p''\in Q\) are such that \(\row(p')=\row(p'')\), for all \(p\in\P\), \(\row(p'\cdot p) = \row(p''\cdot p)\).
    }
    \end{defn}

    Let a closed and consistent observation table \((Q,R,T)\) be given. The Learner constructs a DFA \(\A_{C} = (V_{C},v'_{0},E_{C},\P,\ell_{C})\) from \((Q,R,T)\) by employing Algorithm \ref{algo:dfa_construc}.
    \begin{algorithm}[htbp]
			\caption{\texttt{dfa.construct(\((Q,R,T)\))}} \label{algo:dfa_construc}
		\begin{algorithmic}[1]
			\renewcommand{\algorithmicrequire}{\textbf{Input:}}
			\renewcommand{\algorithmicensure}{\textbf{Output:}}
			
			\REQUIRE An observation table \((Q,R,T)\).
            \ENSURE An automaton, \(\A_{C} = (V_{C},v'_{0},E_{C},\P,\ell_{C})\).
			
			\STATE Set \(m=\) the number of distinct \(\row(\tilde{p})\neq 0\), \(\tilde{p}\in Q\).
            \STATE Set \(v'_{i} = \row(\tilde{p})\), \(i=1,2,\ldots,m\).
            \STATE Set \(V_{C} = \{v'_{i},\:i=1,2,\ldots,m\}\) and \(v_{0}' = \row(\lambda)\).
            \STATE Set \(E_{C} = \{(v'_{i},v'_{j})\:\text{for each}\:p\in\P\:\text{satisfying}\:v'_{i} = \row(p'),\:v'_{j} = \row(p''),\:\row(p'\cdot p)=\row(p''),\:v'_{i},v'_{j}\in V_{C},\:i,j=1,2,\ldots,m\}\).
            \STATE Set \(\ell_{C}(v'_{i},v'_{j}) = p\in\P\) such that \(v'_{i}=\row(p')\), \(v'_{j} = \row(p'')\) and \(\row(p'\cdot p)=\row(p'')\), \(v'_{i},v'_{j}\in V_{C}\), \(i,j=1,2,\ldots,m\).
            \STATE Return \(\A_{C}\).
		\end{algorithmic}
	\end{algorithm}
    Notice that the construction of an automaton in Algorithm \ref{algo:dfa_construc} excludes the rows of \((Q,R,T)\) with all \(0\) entries. Consequently, a \(v\in V\) does not have an outgoing edge labelled with \(p\in\P\) such that \(v = \row(p')\neq 0\), \(\row(p'\cdot p) = 0\), \(p'\in Q\). By construction, \(\A_{C}\) is directed and labelled. In addition,
    \begin{proposition}
    \label{prop:connectivity}
        An automaton \(\A_{C} = (V_{C},v'_{0},E_{C},\P,\ell_{C})\) obtained from Algorithm \ref{algo:dfa_construc} is strongly connected.
    \end{proposition}
    \begin{proof}
        It suffices to show that for each \(p'\in Q\) such that \(\row(p')\neq 0\), there exist \(p''\in Q\) and \(p\in\P\) such that the following conditions are true: \(\row(p'')\neq 0\) and \(\row(p''\cdot p)=\row(p')\).

        Recall that \((Q,R,T)\) employed to construct \(\A_{C}\), is closed and consistent. Let \(p'=p'_{1}p'_{2}\cdots p'_{n}\). Then \(\row(p')=\row(p'_{1}p'_{2}\cdots p'_{n})=\row(p'_{1}p'_{2}\cdots p'_{n-1}\cdot p'_{n})\). Since by definition of an observation table \(Q\) is a prefix-closed set, the sequence of subsystems indices \(p'_{1}p'_{2}\cdots p'_{n-1}\in Q\). Consequently, there exist \(p'' = p'_{1}p'_{2}\cdots p'_{n-1}\) and \(p=p'_{n}\) such that \(\row(p''\cdot p) = \row(p')\).

        It remains to show that \(\row(p'')\neq 0\). It is immediate that if the sequence \(p'_{1}p'_{2}\cdots p'_{n-1}\) is not admissible on \(\A\), then the sequence \(p'_{1}p'_{2}\cdots p'_{n-1}\cdot p'_{n}\) cannot be admissible on \(\A\). The assertion of Proposition \ref{prop:connectivity} follows at once.
    \end{proof}
    We define
    \begin{defn}
    \label{d:lang}
    \rm{
        The \emph{language of \(\A_{C}\)} is the set
        \begin{align*}
            L_{\A_{C}} = \{\tilde{p}_{n}:\tilde{p}_{n}\:\text{is admissible on}\:\A_{C},\:n=1,2,\ldots,M\}\cup\{\lambda\},
        \end{align*}
        where the admissibility of \(\tilde{p}_{n}\) on \(\A_{C}\) corresponds to the existence of a path \(\overline{v}'_{0},\overline{v}'_{1},\ldots,\overline{v}'_{n}\) on \(\A_{C}\) such that \(\overline{v}'_{0} = v'_{0}\), \((\overline{v}'_{k-1},\overline{v}'_{k})\in E\) and \(\ell(\overline{v}'_{k-1},\overline{v}'_{k})=\tilde{p}_{n}(k)\), \(k=1,2,\ldots,n\).
    }
    \end{defn}

    Algorithm \ref{algo:match} matches the languages \(L_{\A}\) and \(L_{\A_{C}}\), and outputs a counter-example, if exists.
     \begin{algorithm}[htbp]
			\caption{\texttt{language.match(\(\A_{C}\))}} \label{algo:match}
		\begin{algorithmic}[1]
			\renewcommand{\algorithmicrequire}{\textbf{Input:}}
			\renewcommand{\algorithmicensure}{\textbf{Output:}}
			
			\REQUIRE A conjectured automaton, \(\A_{C}\).
            		\ENSURE A sequence of subsystems indices, \(\tilde{p}\in\P^*\), if \(\tilde{p}\in L_{\A}\) but \(\tilde{p}\in L_{\A_{C}}\), or vice-versa.
		
			\FOR {\(n=,2,\ldots,M\)}
				\FOR {each \(\tilde{p}_{n}\in\P^*\)}
					\STATE Compute \(\M_{\A}(\tilde{p}_{n})\).
					\IF {\(\tilde{p}_{n}\in L_{\A}\) but \(\tilde{p}_{n}\notin L_{\A_{C}}\), or vice-versa}
						\STATE Return \(\tilde{p}_{n}\) and stop.
					\ENDIF
				\ENDFOR
			\ENDFOR
			\STATE Return \(\emptyset\).
		
		\end{algorithmic}
	\end{algorithm}

    \begin{algorithm}[htbp]
			\caption{Learning a minimal automaton whose language is \(L_{\A}\)} \label{algo:learning}
		\begin{algorithmic}[1]
			\renewcommand{\algorithmicrequire}{\textbf{Input:}}
			\renewcommand{\algorithmicensure}{\textbf{Output:}}
			
			\REQUIRE The number of subsystems, \(N\), the dimension of the subsystems, \(d\), the maximum length, \(M\), of the elements of \(L_{\A}\), and a gray-box simulation model, \(\M_{\A}\), of a switched system \eqref{e:swsys}.	
            \ENSURE The restriction automaton, \(\A = (V,v_{0},E,\P,\ell)\), for \eqref{e:swsys}.
			
			\STATE Construct \(\P^*\).
            \STATE Set \(Q = R = \{\lambda\}\).
            \STATE Compute \(\M_{\A}(\lambda)\) and \(\M_{\A}(p)\) for all \(p\in\P\).
            \STATE Construct an observation table \((Q,R,T)\).\label{step:repeat}
            \WHILE {\((Q,R,T)\) is not closed or not consistent}
                \IF {\((Q,R,T)\) is not closed}
                    \STATE Find \(\tilde{p}\in Q\) and \(p\in\P\) such that \(\row(\tilde{p}\cdot p)\neq \row(p')\) for all \(p'\in Q\).
                    \STATE Set \(Q = Q\cup\{\tilde{p}\cdot p\}\).
                    \STATE Extend \(T\) to \((Q\cup Q\cdot\P)\cdot R\) by computing \(\M_{\A}(\tilde{p})\) for all \(\tilde{p}\in(Q\cup Q\cdot\P)\cdot R\).
                \ENDIF
                \IF {\((Q,R,T)\) is not consistent}
                    \STATE Find \(p',p''\in Q\), \(p\in\P\) and \(\tilde{p}\in R\) such that \(\row(p')=\row(p'')\) and \(T(p'\cdot p\cdot\tilde{p})\neq T(p''\cdot p\cdot\tilde{p})\).
                    \STATE Set \(R = R\cup\{\tilde{p}\cdot p\}\).
                    \STATE Extend \(T\) to \((Q\cup Q\cdot\P)\cdot R\) by computing \(\M_{\A}(\tilde{p})\) for all \(\tilde{p}\in(Q\cup Q\cdot\P)\cdot R\)..
                \ENDIF
            \ENDWHILE
           \STATE Set \(\A_{C} =\)\texttt{dfa.construct(\((Q,R,T)\))}.
           \STATE Set \(\tilde{p}=\)\texttt{language.match(\(\A_{C}\))}.
           \IF {\(\tilde{p}\neq\emptyset\)}
                        \STATE Set \(Q = Q\cup\{\tilde{p}\}\cup\{\tilde{p}_{1},\tilde{p}_{1}\tilde{p}_{2},\tilde{p}_{1}
                        \tilde{p}_{2}\cdots\tilde{p}_{n-1}\}\).
                        \STATE Extend \(T\) to \((Q\cup Q\cdot\P)\cdot R\).
                        \STATE Go to Step \ref{step:repeat}.
           \ELSE
            		\STATE Output \(\A^* = \A_{C}\) and stop.
	   \ENDIF
		\end{algorithmic}
	\end{algorithm}

    The learning technique employed in Algorithm \ref{algo:learning} relies on the \(L^*\)-algorithm. In the \(L^*\)-algorithm, the Learner learns a DFA that accepts a certain language \(L\), with the aid of an oracle called the \emph{minimally adequate teacher} (MAT). A DFA \(\D\) under consideration in \cite{Angluin1987} is a tuple \((Q,q_{0},\Sigma,F,\delta)\), where \(Q\) is a finite set of nodes, \(q_{0}\in Q\) is the unique initial node, \(\Sigma\) is a finite set of alphabets, \(F\subseteq Q\) is a finite set of accepting (or final) nodes, and \(\delta:Q\times\Sigma\to Q\) is the node transition function. The language of \(\D\) is the set of all finite words (strings of alphabets) such that the DFA reaches a final node on reading them, i.e., a word \(w = w_{1}w_{2}\cdots w_{m}\), \(w_{k}\in\Sigma\), \(k=1,2,\ldots,m\), belongs to the language of \(\D\), if \(\delta(\cdots(\delta(\delta(q_{0},w_{1}),w_{2}),\cdots,w_{m})\in F\). The MAT knows \(L\) and answers two types of queries by the Learner: \emph{membership queries}, i.e., whether or not a given word belongs to \(L\), and \emph{equivalence queries}, i.e., whether a DFA conjectured by the Learner is correct or not. If the language of the conjectured DFA differs from \(L\), then the MAT responds to an equivalence query with a counter-example, which is a word that is misclassified by the conjectured DFA. The \(L^*\)-algorithm considers \(\D\) to be complete in the sense that there is a valid transition corresponding to every pair of node and alphabet. The class of restriction automata considered in this paper differs structurally from the DFA's considered in \cite{Angluin1987} in the following ways: (a) \(\A\) has \(0\)-many accepting nodes, and (b) \(\A\) is not necessarily complete in the sense that there may exist \(v\in V\) and \(p\in\P\) such that \(\ell(v,v')\neq p\) for any \(v'\in V\) with \((v,v')\in E\). In Algorithm \ref{algo:learning} we modify the \(L^*\)-algorithm to cater to learning of \(\A\). Loosely speaking, \(\M_{\A}\) plays the role of a MAT. The Learner initializes the sets \(Q\) and \(R\) to be \(\lambda\), and updates them by checking for the existence of the elements in \((Q\cup Q\cdot\P)\cdot R\) in \(L_{\A}\), until a closed and consistent observation table \((Q,R,T)\) is obtained. She conjectures an automaton, \(\A_{C}\), with the data in \((Q,R,T)\). Then the existence of every element \(\tilde{p}\in\P^*\) in \(L_{\A}\) and \(L_{\A_{C}}\) are matched. If a counter-example is obtained, then the Learner first updates \(Q\), and then updates \((Q,R,T)\) until the latter becomes closed and consistent. The algorithm terminates when a counter-example is not found. The latest version of \(\A_{C}\) is output as \(\A^*\). The procedures for constructing a closed and consistent observation table and incorporating a counter-example into the observation table employed in Algorithm \ref{algo:learning} are the same with the procedures employed in the \(L^*\)-algorithm for these tasks. However, our construction of an automaton from a closed and consistent observation table in Algorithm \ref{algo:dfa_construc} differs from the construction of a DFA employed in the \(L^*\)-algorithm. The \(0\)-entries in an observation table are utilized to designate the non-accepting nodes of a DFA in the \(L^*\)-algorithm, while we exclude the rows of \((Q,R,T)\) with all \(0\) entries to restrict \(\A_{C}\) to the admissible sequences of subsystems indices. Moreover, for learning a DFA, the \(L^*\)-algorithm does not require the maximum length of the elements of a regular language to be specified prior to its application. The MAT confirms correctness of the conjectured DFA by providing, if exists, a counter-example. In our setting, the Learner herself checks for correctness of the conjectured automaton by searching for a possible counter-example in Algorithm \ref{algo:match}, and the knowledge of \(M\) is utilized to this end.

    \begin{proposition}
    \label{prop:mainres}
        Algorithm \ref{algo:learning} terminates in bounded time.
    \end{proposition}
    \begin{proof}
     From \cite[Theorem 6]{Angluin1987} it follows that the running time of Algorithm \ref{algo:learning} is bounded by a polynomial in the number of nodes of \(\A\) and the length of the longest counter-example obtained. We have that the number of nodes of \(\A\) is \(\abs{V}\). Moreover, by construction of a counter-example, its maximum possible length is \(M\). Consequently, Algorithm \ref{algo:learning} terminates in bounded time.
   \end{proof}
   \begin{proposition}
    \label{prop:mainres123}
        \(\A^*\) is the minimal automaton whose language is \(L_{\A}\).
    \end{proposition}
    \begin{proof}
     From \cite[Theorem 6]{Angluin1987} it follows that Algorithm \ref{algo:learning} outputs an automaton, \(\A^*\), isomorphic to the minimal automaton whose language is \(L_{\A}\). The assertion of Proposition \ref{prop:mainres123} follows at once.\footnote{Recall that two automata are \emph{equivalent} if they accept the same language, and two isomorphic automata are equivalent. The reader is referred to \cite[Chapter 4]{Hopcroft} for a detailed discussion on isomorphism and equivalence of automata.}
    \end{proof}

    We next combine Algorithms \ref{algo:model_learn} and \ref{algo:learning} towards learning a switched system \eqref{e:swsys}.
\subsection{Learning a switched system \eqref{e:swsys} from \(\M\)}
\label{ss:swsys_learning}
    \begin{algorithm}[htbp]
			\caption{Learning a switched system \eqref{e:swsys}} \label{algo:learning_swsys}
		\begin{algorithmic}[1]
			\renewcommand{\algorithmicrequire}{\textbf{Input:}}
			\renewcommand{\algorithmicensure}{\textbf{Output:}}
			
			\REQUIRE The total number of subsystems, \(N\), the dimension of the subsystems, \(d\), the order of the component polynomials, \(m\), the maximum length, \(M\), of the elements of \(L_{\A}\), and a gray-box simulation model, \(\M\), of the switched system \eqref{e:swsys}.
            \ENSURE The functions \(f_{p}\), \(p\in\P\) and a minimal automaton whose language is \(L_{\A}\).

            \STATE Apply Algorithm \ref{algo:model_learn} to learn \(f_{p}\), \(p\in\P\).
            \STATE Apply Algorithm \ref{algo:learning} to learn a minimal automaton whose language is \(L_{\A}\).
        \end{algorithmic}
	\end{algorithm}

    Given the number of subsystems, the dimension of the subsystems, the order of the component polynomials, the maximum length of the elements in the language of the restriction automaton, and a simulation model from which certain information about state trajectories of the individual subsystems and existence of finite sequences of subsystems indices in the language of the restriction automaton can be obtained, Algorithm \ref{algo:learning_swsys} infers a switched system \eqref{e:swsys} by learning the functions \(f_{p}\), \(p\in\P\) and a minimal automaton whose language is \(L_{\A}\). Algorithms \ref{algo:model_learn} and \ref{algo:learning} are employed to learn \(f_{p}\), \(p\in\P\) from \(\M_{\P}\) and a minimal automaton that accepts the language \(L_{\A}\) from \(\M_{\A}\), respectively. From Propositions \ref{prop:bounded_time} and \ref{prop:mainres}, it follows that Algorithm \ref{algo:learning_swsys} terminates in bounded time. Moreover, the correctness of the outputs of Algorithm \ref{algo:learning_swsys} is asserted by Propositions \ref{prop:unique_soln} and \ref{prop:mainres123}.
    \begin{example}
    \label{ex:numex}
    \rm{
    Consider a switched system \eqref{e:swsys} with \(\P = \{1,2,3\}\),
    \begin{align*}
        f_{1}(x) &= \pmat{-0.0625 + 0.125x_{1} -0.25x_{1}^{2} +0.5x_{1}^{3}\\
                          0.0625 - 0.125x_{2} + 0.25x_{2}^{2} + 0.5x_{2}^{3}\\
                          0.0625 - 0.125x_{3} + 0.25x_{3}^{2} - 0.5x_{3}^3},\\
        f_{2}(x) &= \pmat{0.3x_{1}^{3}-0.27x_{1}+0.81\\
                         -0.3x_{2}^{3}+0.27x_{2}+0.81\\
                         0.3x_{3}^{3}+0.81x_{3}},\\
        f_{3}(x) &= \pmat{-2x_{1}^{3}+4x_{1}^{2}\\
                          -4x_{2}^{3}+8x_{2}^{2}\\
                          0.5x_{3}^{3}+0.25x_{3}^{2}-0.625}.
    \end{align*}
     Let an admissible switching signal obey the following restrictions: a) the system starts operating at subsystem \(1\), b) the admissible switches between the subsystems are \(1\to 1\), \(1\to 2\), \(1\to 3\), \(2\to 1\), \(3\to 1\), and c) the admissible dwell times on subsystems \(2\) and \(3\) are one unit of time. The corresponding restriction automaton \(\A\) contains: \(V = \{v_{0},v_{1}\}\), \(E = \{(v_{0},v_{1}),(v_{1},v_{1}),(v_{1}^{(1)},v_{0}^{(1)}),(v_{1}^{(2)},v_{0}^{(2)})\}\), \(\ell(v_{0},v_{1}) = 1\), \(\ell(v_{1},v_{1}) = 1\), \(\ell(v_{1}^{(1)},v_{0}^{(1)}) = 2\), \(\ell(v_{1}^{(2)},v_{0}^{(2)}) = 3\), where \(v_{0}^{(i)} = v_{0}\), \(v_{1}^{(i)} = v_{1}\), \(i=1,2\). A pictorial representation of \(\A\) is shown in Figure \ref{fig:ex2_automat1}.
       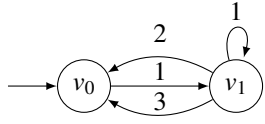
\begin{figure}[htbp]
    \centering
        \begin{tikzpicture}[every path/.style={>=latex},base node/.style={draw,circle}]
            \node[base node]            (a) at (-1,0)  { $v_{0}$ };
            \node[base node]            (b) at (1,0)  { $v_{1}$ };

             \draw[->] (b) edge[loop above] (b);
             \draw[->] (a) edge (b);
             \draw[->] (b) edge[bend right] (a);
             \draw[->] (b) edge[bend left] (a);
             \draw[->] (-2,0) -- (a);
             \node (s) at (0,0.2) {$1$};
             \node (s) at (1,1) {$1$};
             \node (s) at (0,0.7) {$2$};
             \node (s) at (0,-0.2) {$3$};
         \end{tikzpicture}
        \caption{Restriction automaton} \label{fig:ex2_automat1}
    \end{figure}

     Suppose that the Learner knows \(N = 3\), \(d = 3\), \(m=3\), \(M = 100\), and has access to a simulation model, \(\M\), of \eqref{e:swsys}. She applies Algorithm \ref{algo:learning_swsys} to infer \eqref{e:swsys}. The learning process involves two steps:\\
     {Step 1}: Learning the functions \(f_{p}\), \(p\in\P\) from \(\M_{\P}\) by applying Algorithm \ref{algo:model_learn}.\\
     I. For \(p=1\), the Learner sets \(\overline{x}_{0} = \pmat{0 & 0 & 0}^\top\), \(\overline{x}_{1} = \pmat{1 & 1 & 1}^\top\), \(\overline{x}_{2} = \pmat{2 & 2 & 2}^\top\), \(\overline{x}_{3} = \pmat{3 & 3 & 3}^\top\), inputs \((p,\overline{x}_{k})\), \(k=0,1,2,3\) to \(\M_{\P}\) and obtains \(x'_{0} = \pmat{-0.0625 & 0.0625 & 0.0625}^\top\), \(x'_{1} =\pmat{0.3125 & 0.6875 & -0.3125}^\top\), \(x'_{2} = \pmat{3.1875 & 4.8125 & -3.1875}^\top\),\\ \(x'_{3} = \pmat{11.5625 & 15.4375 & -11.5625}^\top\). She assigns \(a_{1,10} = -0.0625\), \(a_{1,20} = 0.0625\), \(a_{1,30} = 0.0625\), and then solves \eqref{e:sys_lin} with \(i=1,2,3\), \(m=3\). The following values are obtained: \(a_{1,11} = 0.125\), \(a_{1,12} = -0.25\), \(a_{1,13} = 0.5\), \(a_{1,21} = -0.125\), \(a_{1,22} = 0.25\), \(a_{1,23} = 0.5\), \(a_{1,31} = -0.125\), \(a_{1,32} = 0.25\), \(a_{1,33} = -0.5\). The function \(f_{1}\) is output as \(\pmat{-0.0625 + 0.125x_{1} -0.25x_{1}^{2} +0.5x_{1}^{3}\\
                          0.0625 - 0.125x_{2} + 0.25x_{2}^{2} + 0.5x_{2}^{3}\\
                          0.0625 - 0.125x_{3} + 0.25x_{3}^{2} - 0.5x_{3}^3}\).\\
     II. For \(p=2\), the Learner sets \(\overline{x}_{0}\), \(\overline{x}_{1}\), \(\overline{x}_{2}\), \(\overline{x}_{3}\) as described in I. above, inputs \((p,\overline{x}_{k})\), \(k=0,1,\ldots,m\) to \(\M_{\P}\), and obtains \(x'_{0} = \pmat{0.81 & 0.81 & 0}^\top\), \(x'_{1} = \pmat{0.84 & 0.78 & 1.11}^\top\), \(x'_{2} = \pmat{2.67 & -1.05 & 4.02}^\top\), \(x'_{3} = \pmat{8.1 & -6.48 & 10.53}^\top\). She assigns \(a_{2,10} = 0.81\), \(a_{2,20} = 0.81\), \(a_{2,30} = 0\), and then solves \eqref{e:sys_lin} with \(i=1,2,3\), \(m=3\). The following values are obtained: \(a_{2,11} = -0.27\), \(a_{2,12} = 0\), \(a_{2,13} = 0.3\), \(a_{2,21} = 0.27\), \(a_{2,22} = 0\), \(a_{2,23} = -0.3\), \(a_{2,31} = 0.81\), \(a_{2,32} = 0\), \(a_{2,33} = 0.3\). The function \(f_{2}\) is output as \(\pmat{0.3x_{1}^{3}-0.27x_{1}+0.81\\
                         -0.3x_{2}^{3}+0.27x_{2}+0.81\\
                         0.3x_{3}^{3}+0.81x_{3}}\).\\
    III. For \(p=3\), the Learner obtains \(x'_{0} = \pmat{0.81 & 0.81 & 0}^\top\), \(x'_{1} = \pmat{0.84 & 0.78 & 1.11}^\top\), \(x'_{2} = \pmat{2.67 & -1.05 & 4.02}^\top\), \(x'_{3} = \pmat{8.1 & -6.48 & 10.53}^\top\) by providing \((p,\overline{x}_{k})\), \(k=0,1,2,3\) to \(\M_{\P}\) as input, where \(\overline{x}_{k}\), \(k=0,1,2,3\) are as described in I. above. She assigns \(a_{3,10} = 0\), \(a_{3,20} = 0\), \(a_{3,30} = -0.625\), and then solves \eqref{e:sys_lin} with \(i=1,2,3\), \(m=3\). The following values are obtained: \(a_{3,11} = 0\), \(a_{3,12} = 4\), \(a_{3,13} = -2\), \(a_{3,21} = 0\), \(a_{3,22} = 8\), \(a_{3,23} = -4\), \(a_{3,31} = 0\), \(a_{3,32} = 0.25\), \(a_{3,33} = 5\). The function \(f_{3}\) is output as \(\pmat{-2x_{1}^{3}+4x_{1}^{2}\\
                          -4x_{2}^{3}+8x_{2}^{2}\\
                          0.5x_{3}^{3}+0.25x_{3}^{2}-0.625}\).\\
    {Step 2}: Learning a minimal automaton whose language is \(L_{\A}\) from \(\M_{\A}\) by applying Algorithm \ref{algo:learning}.\\
    I. \(\P^*\) is constructed.\\
    II. \(Q = R = \{\lambda\}\) are initialized and an observation table \(T_{1} = (Q,R,T)\) shown in Table \ref{tab:ex2_tab1} is constructed. \(T_{1}\) is not closed but consistent. The Learner updates \(Q = Q\cup\{2\}\), and extends \(T\) to \((Q\cup Q\cdot\P)\cdot R\) to construct the observation table \(T_{2} = (Q,R,T)\) shown in Table \ref{tab:ex2_tab2}. \(T_{2}\) is closed and consistent. The corresponding automaton, \(\A_{C}\), contains: \(V_{C} = \{v'_{0}\}\), \(E_{C} = \{(v'_{0},v'_{0})\}\), \(\ell_{C}(v'_{0},v'_{0}) = 1\). It is shown pictorially in Figure \ref{fig:ex2_automat2}. The sequence \(\tilde{p} = 12\) is a counter-example.\\
    III. The Learner updates \(Q = Q\cup\{1,12\}\), and extends \(T\) to \((Q\cup Q\cdot\P)\cdot R\) to construct the observation table \(T_{3} = (Q,R,T)\) shown in Table \ref{tab:ex2_tab3}. \(T_{3}\) is closed but not consistent. The Learner picks \(p' = 1\), \(p'' = 12\), \(p=2\), \(\tilde{p} = \lambda\), updates \(R = R\cup\{2\}\), and extends \(T\) to \((Q\cup Q\cdot\P)\cdot R\) to construct the observation table \(T_{4} = (Q,R,T)\) shown in Table \ref{tab:ex2_tab4}. \(T_{4}\) is closed and consistent. The corresponding automaton, \(\A_{C}\), contains: \(V_{C} = \{v'_{0},v'_{1}\}\), \(E_{C} = \{(v'_{0},v'_{1}),(v'_{1},v'_{1}),(v'^{(1)}_{1},v'^{(1)}_{0}),(v'^{(2)}_{1},v'^{(2)}_{0})\}\), \(\ell_{C}(v'_{0},v'_{1}) = 1\), \(\ell_{C}(v'_{1},v'_{1}) = 1\), \(\ell_{C}(v'^{(1)}_{1},v'^{(1)}_{0}) = 2\), \(\ell_{C}(v'^{(2)}_{1},v'^{(2)}_{0}) = 3\), where \(v'^{(i)}_{0} = v'_{0}\), \(v'^{(i)}_{1} = v'_{1}\), \(i=1,2\). Figure \ref{fig:ex2_automat2} shows \(\A_{C}\) pictorially. No \(\tilde{p}\in\P^*\) is a counter-example, and Algorithm \ref{algo:learning} outputs \(\A^* = \A_{C}\).
    \begin{table}[htbp]
	\centering
	\begin{tabular}{|c |c|}
        \hline
        \(T_{1}\) & \(\lambda\)\\
		\hline
        \(\lambda\) & \(1\)\\
        \hline
        \(1\) & \(1\)\\
        \(2\) & \(0\)\\
        \(3\) & \(0\)\\
        \hline
	\end{tabular}
	\caption{Observation table \(T_{1}\)}\label{tab:ex2_tab1}
	\end{table}
    \begin{table}[htbp]
	\centering
	\begin{tabular}{|c |c|}
        \hline
        \(T_{1}\) & \(\lambda\)\\
		\hline
        \(\lambda\) & \(1\)\\
        \(2\) & \(0\)\\
        \hline
        \(1\) & \(1\)\\
        \(3\) & \(0\)\\
        \(21\) & \(0\)\\
        \(22\) & \(0\)\\
        \(23\) & \(0\)\\
        \hline
	\end{tabular}
	\caption{Observation table \(T_{2}\)}\label{tab:ex2_tab2}
	\end{table}
    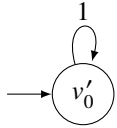
\begin{figure}[htbp]
    \centering
        \begin{tikzpicture}[every path/.style={>=latex},base node/.style={draw,circle}]
            \node[base node]            (a) at (0,0)  { $v'_{0}$ };

             \draw[->] (a) edge[loop above] (a);
            \node (s) at (0,1.1) {$1$};
            \draw[->] (-1,0) -- (a);
         \end{tikzpicture}
        \caption{Automaton corresponding to \(T_{2}\)} \label{fig:ex2_automat2}
    \end{figure}
     \begin{table}[htbp]
	\centering
	\begin{tabular}{|c |c|}
        \hline
        \(T_{3}\) & \(\lambda\)\\
		\hline
        \(\lambda\) & \(1\)\\
        \(1\) & \(1\)\\
        \(2\) & \(0\)\\
        \(12\) & \(1\)\\
        \hline
        \(3\) & \(0\)\\
        \(11\) & \(1\)\\
        \(13\) & \(1\)\\
        \(21\) & \(0\)\\
        \(22\) & \(0\)\\
        \(23\) & \(0\)\\
        \(121\) & \(1\)\\
        \(122\) & \(0\)\\
        \(123\) & \(0\)\\
        \hline
	\end{tabular}
	\caption{Observation table \(T_{3}\)}\label{tab:ex2_tab3}
	\end{table}
     \begin{table}[htbp]
	\centering
	\begin{tabular}{|c |c| c|}
        \hline
        \(T_{3}\) & \(\lambda\) & \(2\)\\
		\hline
        \(\lambda\) & \(1\) & \(0\)\\
        \(1\) & \(1\) & \(1\)\\
        \(2\) & \(0\) & \(0\)\\
        \(12\) & \(1\) & \(0\)\\
        \hline
        \(3\) & \(0\) & \(0\)\\
        \(11\) & \(1\)& \(1\)\\
        \(13\) & \(1\)& \(0\)\\
        \(21\) & \(0\)& \(0\)\\
        \(22\) & \(0\)& \(0\)\\
        \(23\) & \(0\)& \(0\)\\
        \(121\) & \(1\)& \(1\)\\
        \(122\) & \(0\)& \(0\)\\
        \(123\) & \(0\)& \(0\)\\
        \hline
	\end{tabular}
	\caption{Observation table \(T_{4}\)}\label{tab:ex2_tab4}
	\end{table}
    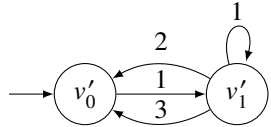
\begin{figure}[htbp]
    \centering
        \begin{tikzpicture}[every path/.style={>=latex},base node/.style={draw,circle}]
            \node[base node]            (a) at (-1,0)  { $v'_{0}$ };
            \node[base node]            (b) at (1,0)  { $v'_{1}$ };

             \draw[->] (b) edge[loop above] (b);
             \draw[->] (a) edge (b);
             \draw[->] (b) edge[bend right] (a);
             \draw[->] (b) edge[bend left] (a);
             \draw[->] (-2,0) -- (a);
             \node (s) at (0,0.2) {$1$};
             \node (s) at (1,1.1) {$1$};
             \node (s) at (0,0.7) {$2$};
             \node (s) at (0,-0.2) {$3$};
         \end{tikzpicture}
        \caption{Automaton learnt from Algorithm \ref{algo:learning}} \label{fig:ex2_automat3}
    \end{figure}
    }
    \end{example}

\section{Conclusion}
\label{s:concln}
    In this paper we presented an algorithm to learn discrete-time switched systems whose subsystems dynamics are governed by sets of scalar polynomials and admissible switching signals are governed by restriction automata. Given the number of subsystems, the dimension of the subsystems, the order of the component polynomials, the maximum length of the elements in the language of the restriction automaton, and a gray-box simulation model of the switched system, our algorithm learns the coefficients of the component polynomials and a minimal automaton that accepts the language of the restriction automata correctly, in bounded time. Our learning technique relies on linear algebraic tools and the \(L^*\)-algorithm from machine learning literature. A next natural research direction is the design of active learning techniques for the identification of larger classes of switched systems, e.g., whose subsystems dynamics are nonlinear but not restricted to the functions of scalar polynomial considered in this paper, the restrictions on the sets of admissible switching signals are non-deterministic, etc. This matter is currently under investigation and the findings will be reported elsewhere.






\end{document}